\documentclass[11pt]{article}
\usepackage[a4paper,left=25mm,right=25mm,top=25mm, bottom=25mm]{geometry}
\usepackage[usenames]{color}
\usepackage{a4wide}
\usepackage{iftex}
\usepackage{amsmath}
\usepackage{amssymb}
\usepackage{tikz}
\usepackage{bm}

\usepackage{natbib}
\usepackage{hyperref}

\usepackage{epsfig}
\usepackage{amsfonts}
\usepackage{color}

\newtheorem{theorem}{Theorem}
\newtheorem{definition}{Definition}
\newtheorem{proposition}{Proposition}

\newtheorem{corollary}{Corollary}

\newtheorem{example}{Example}
\newenvironment{proof}{\par\noindent\underline{Proof}: }{\hfill$\square$\par\vskip10pt}

\newcommand{\med}{\mathit med}
\renewcommand{\baselinestretch}{1.3}\normalsize
\begin{document}

\title{Anonymity and strategy-proofness on a domain of single-peaked and single-dipped preferences}


\author{Oihane Gallo\thanks{\footnotesize{Department of Mathematical Economics, Finance, and Actuarial Sciences, University of Barcelona, Spain.   \texttt{Email:\,oihane.gallo@ub.edu}. Financial support from the Swiss National Science Foundation (SNSF) through project 100018$\_$192583 and from the Spanish Ministry of Economy and Competitiveness, through project PID2021-127119NB-I00 (funded by MCIN/AEI/
10.13039/501100011033 and by “ERDF A way of making Europe”) is gratefully acknowledged. I also thank Bettina Klaus and William Thomson for helpful comments and suggestions.}}}
\thispagestyle{empty}
\date{\today}
\maketitle
\begin{abstract}\vspace{0.1cm}
\noindent We analyze the problem of locating a public facility on a line in a society where agents have  single-peaked or single-dipped preferences. We consider the domain analyzed in \cite{alcalde2024strategy}, where the type of preference of each agent is public information, but the location of her peak/dip as well as the rest of the preference are unknown. We characterize all strategy-proof and type-anonymous rules on this domain. Building on existing results, we provide a two-step characterization: first, the median between the peaks and a collection of fixed locations is computed \cite[]{moulin1980strategy}, resulting in either a single alternative or a pair of contiguous alternatives. If the outcome of the median is a pair, we apply a ``double-quota majority method" in the second step to choose between the alternatives in the pair \cite[]{moulin1983strategy}. We also show the additional restrictions implied by type-anonymity on the strategy-proof rules characterized by \cite{alcalde2024strategy}. Finally, we show the equivalence of the two characterizations. \vspace{0.15cm}
\end{abstract}

\vspace{0.15cm}

\noindent \textit{JEL-Numbers:} D70, D71, D79.

\newpage
\renewcommand{\baselinestretch}{1.5}\normalsize
\section*{Introduction} 

\noindent A new public facility needs to be located in your city, but where? City hall has decided to consider the preferences of the citizens to take the final decision. The aggregation of these preferences is made following two requirements: \emph{strategy-proofness} (no agent has an incentive to misrepresent her preferences) and \emph{anonymity} (all agents have equal power in the decision-making process).\medskip

\noindent It is well-known from the Gibbard-Satterthwaite theorem \cite[]{gibbard1973manipulation,satterthwaite1975strategy} that within the universal domain, where all possible preferences are feasible, no social choice rule with more than two alternatives in its range can simultaneously satisfy two crucial principles: strategy-proofness and non-dictatorship (no single agent can dictate the group's decision).\medskip

\noindent Nevertheless, many social and economic situations naturally lead to restricted preference domains such as determining the location of a public facility. Depending on the nature of the facility to be located, it may be natural for agents to have specific types of preferences. For instance, if the facility has a positive impact on the area - such as a museum that attracts tourists that consume in local shops and restaurants - agents may prefer it to be located closer to their own location. This situation induces single-peaked preferences, where each agent has an ideal point, and the further away the facility is from that point, the worse off the agent is. \cite{black1948decisions, black1948rationale} was the first to discuss single-peaked preferences and demonstrated that the median voter rule, which selects the median of the declared peaks, is strategy-proof and selects the Condorcet winner.\footnote{A Condorcet winner is a candidate who would receive the support of more than half of the electorate in any one-to-one election against each of the other candidates.} Later, \cite{moulin1980strategy} and \cite{barbera1994characterization} characterized all strategy-proof rules within this domain as ``generalized median voter rules." In contrast, if the facility has a negative impact on the neighborhood - such as a nuclear plant that produces radioactive waste and poses health risks - agents may prefer it to be located far away from their own location. This situation induces single-dipped preferences, where each agent has a least preferred point, and the farther away the facility is from that point, the better off the agent is. \cite{barbera2012domains} and \cite{manjunath2014efficient} established that within this domain, all strategy-proof rules have a range of at most two alternatives.\medskip

\noindent Mixed domains of these two types of preferences have been analyzed in the literature. \cite{berga2000maximal} and \cite{achuthankutty2018dictatorship} demonstrated that the Gibbard-Satterthwaite result still holds when the set of admissible preferences for each agent includes all single-peaked and all single-dipped preferences. Therefore, further constraints on this mixed preference domain are necessary to escape the Gibbard-Satterthwaite dictatorship result. For instance, \cite{alcalde2018strategy} characterized all strategy-proof rules in a mixed domain where the peak/dip of each agent is public information, while both the type of preference and the rest of the preference structure (\emph{i.e.} how the other alternatives are ordered in the preference within the corresponding domain) are unknown.\medskip

\noindent In this paper, we are interested in the location of a public facility, such as a train station, a soccer stadium, or a shopping mall, where agents may have one or the other type of preference. On the one hand, a train station can be useful for someone who commutes to work, leading to single-peaked preferences with the peak at her house or workplace. On the other hand, train stations are noisy, and those agents living or working nearby may have single-dipped preferences with the dip at their house or workplace. Then, society is partitioned into agents with single-peaked preferences and agents with single-dipped preferences. Related domains have been previously analyzed. \cite{thomson2022should} examined a restricted domain with two agents, where one has single-peaked preferences and the other has single-dipped preferences, with both the peak and the dip, which are located at the same point, being public information. The author found that the dictatorship result still prevails in this domain. However, other authors have successfully characterized strategy-proof rules with more than two alternatives in their range for other restricted mixed domains. \cite{feigenbaum2015strategyproof} explored a model where the type of preference of each agent (single-peaked or single-dipped) is known, and the preference of each agent is determined cardinally by the distance between her location and her peak/dip, being the peak/dip private information. More recently, \cite{alcalde2024strategy} examined a mixed domain in which the type of preference of each agent (single-peaked or single-dipped) is public information, while the location of the peak/dip and the rest of the preference are private information as in \cite{thomson2022should}. \medskip

\noindent We focus on the domain introduced in \cite{alcalde2024strategy}. 
In addition to the property of strategy-proofness they analyzed, we emphasize the importance of ensuring that all agents have equal influence over the outcome: We seek decision-making procedures that treat all agents equitably, a principle captured by the concept of anonymity. Anonymity requires that the social choice function returns the same outcome if the preferences of the agents are permuted arbitrarily. However, since agents in our domain have different types of preferences, the classical definition of anonymity cannot be directly applied. We introduce an alternative property, \emph{type-anonymity}, which applies anonymity within types, allowing permutations only among agents with the same type of preferences. We then charaterize all strategy-proof and type-anonymous rules.\medskip

\noindent The first characterization is based on well-known results in the literature (Theorem~\ref{theorem1}): any strategy-proof and type-anonymous rule can be described by means of a two-step procedure as follows: in the first step, the median between the peaks and a fixed collection of locations, which can be single alternatives or pairs of contiguous alternatives, if there exist, is computed. If the median is a single alternative, then it becomes the final outcome. Otherwise, in the second step, a ``double-quota majority method" is applied to choose between the two alternatives of the pair. It is noteworthy that the first step exclusively involves agents with single-peaked preferences. \cite{moulin1980strategy} characterized all strategy-proof and anonymous rules on the single-peaked domain. He worked on the extended real line and stated that the outcome of any strategy-proof and anonymous rule coincides with the median between the agents' peaks and a fixed collection of locations (real values). Even though we cannot directly apply this result as the outcome of the first step may be a pair of alternatives, we find a similar result by imposing certain conditions on the feasible fixed collection of locations. In the second step, we face a binary choice problem. \cite{moulin1983strategy} characterized strategy-proof and anonymous rules for choosing between two alternatives in the case of strict preferences as ``quota majority methods". Given two alternatives $a$ and $b$, a quota majority method is defined by a threshold (quota), which is the minimum support required to implement alternative $a$. If the number of agents who prefer $a$ to $b$ is higher than or equal to that threshold, $a$ is selected; otherwise, the rule selects $b$. In our context, some agents have single-peaked preferences, while others have single-dipped preferences. Therefore, in the second step, two thresholds are defined: one for the agents with single-peaked preferences and another one for the agents with single-dipped preferences. Hence, the left alternative of the pair is chosen if both the number of agents with single-peaked preferences and the number of agents with single-dipped preferences are higher than or equal to their corresponding threshold. Otherwise, the right alternative is chosen.  \medskip

\noindent Regarding the second characterization (Theorem~\ref{theorem2}), building upon the main result in \cite{alcalde2024strategy}, we investigate the subfamily of strategy-proof rules that also satisfy type-anonymity. Our aim is to identify the additional restrictions this new property implies on the previously characterized class of strategy-proof rules. We demonstrate that these requirements are based on the number of agents who agree on the decision at each step, rather than on the identities of these agents. Finally, we show that both characterizations (Theorems~\ref{theorem1} and \ref{theorem2}) are equivalent.\medskip

\noindent The remainder of the paper is organized as follows. The model and notation are introduced in Section~\ref{sec1}. In Section~\ref{sec2}, we provide an overview of existing relevant literature. Sections~\ref{sec3} and \ref{sec4} introduce the two different characterizations of the strategy-proof and type-anonymous rules. In Section~\ref{sec5}, we show the equivalence between both characterizations. Finally, Section~\ref{sec6} concludes.

\section{The basic model and definitions}\label{sec1}

\noindent Let us consider a finite set of agents $N=\{1,\ldots,n\}$ and a set of feasible alternatives $X \subseteq \mathbb{R}$. The set of agents $N$ is divided into two sets $A$ and $D=N\setminus A$. Let $|A|=a \in \mathbb{N}\cup\{0\}$ and $|D|=n-a \in \mathbb{N}\cup\{0\}$. Each agent $i\in N$ has a preference relation $R_i$ over the alternatives in $X$. Formally, $R_i$ is a complete, transitive, and antisymmetric binary relation. Let $P_i$ denote the strict preference relation induced by $R_i$. Let ${\cal R}_{i}$ be the set of admissible preferences of agent $i$, and ${\cal R} = \times_{i \in N} {\cal R}_{i}$ the domain of preferences.\medskip

\noindent A preference $R_i$ is \textbf{single-peaked} if there exists $\rho(R_i) \in X$, called agent $i$'s \emph{peak}, such that for each $x, y \in X$, if $\rho(R_i) \geq x > y$ or $\rho(R_i) \leq x < y$, then $x \, P_i \, y$. Similarly, a preference $R_i$ is \textbf{single-dipped} if there exists $\delta(R_i) \in X$, called agent $i$'s \emph{dip}, such that for each $x, y \in X$, if $\delta(R_i) \geq x > y$ or $\delta(R_i) \leq x < y$, then $y \, P_i \, x$. The set of all single-peaked preferences is denoted by $\mathcal{R}_{sp}$ and the set of all single-dipped preferences is denoted by $\mathcal{R}_{sd}$. For each $i\in A$, ${R}_i\in\mathcal{R}_{sp}$ and for each agent $i\in D$, $ R_i\in\mathcal{R}_{sd}$. \medskip

\noindent A preference profile is a list of preferences $R \equiv (R_i)_{i \in N} \in {\cal R}$. For each $S \subset N$, let ${\cal R}^S = ({\cal R}^i)_{i \in S}$ be the subdomain of ${\cal R}$ restricted to agent set $S$. Given profile $R \in {\cal R}$, subprofiles $R_S \in {\cal R}^S$ and $R_{-S} \in {\cal R}^{N \setminus S}$ are obtained by restricting $R$ to $S$ and $N\setminus S$, respectively.\medskip

\noindent A social choice rule, or simply a \textbf{rule}, on ${\cal R}$ is a function $f: {\cal R} \rightarrow X$. Let $\Omega_f$ denote the range of $f$, i.e., the set of alternatives that appear as the outcome of $f$ for some profiles. Formally, $\Omega_f\equiv\{x\in X: \exists R\in\mathcal{R} \mbox{\ such\ that\ } f(R)=x\}$. We assume that $\Omega_f$ is such that its complementary set $\mathbb{R}\setminus\Omega_f$ is either empty or the union of open sets.\footnote{This is assumed to be able to define the concepts of $\Omega_f$-restricted peaks and dips in the following sentences.} For each $i \in A$ and each $R_i \in {\cal R}^i$, let $p_{\Omega_f}(R_i)$ denote the agent $i$'s most preferred alternative in $\Omega_f$, called \textbf{$\Omega_f$-restricted peak}, at $R_i$, i.e., $p_{\Omega_f}(R_i) \equiv \{x \in \Omega_f \, : \, x \, P_i \, y \mbox{ for each } y \in \Omega_f \setminus \{x\}\}$.  Similarly, for each $i \in D$ and each $R_i \in {\cal R}^i$, let $d_{\Omega_f}(R_i)$ denote the agent $i$'s least preferred alternative in $\Omega_f$, called \textbf{$\Omega_f$-restricted dip}, at $R_i$, i.e., $d_{\Omega_f}(R_i) \equiv \{x \in \Omega_f \, : \, y \, P_i \, x \mbox{ for each } y \in \Omega_f \setminus \{x\}\}$. For simplicity, and since it does not affect the analysis, we simply denote the $\Omega_f$-restricted peaks and dips as $p(R_i)$ and $d(R_i)$, respectively, throughout the paper. That is, $p(R_i)=p_{\Omega_f}(R_i)$ and $d(R_i)=d_{\Omega_f}(R_i)$. Let $\Omega_f^a$ denote any vector of $a$ components taking values in the range of the rule $f$, i.e., $\Omega_f^a \equiv \{(x^1,\ldots,x^{a}) \in X^a \, : \, x^1,\ldots, x^a \in \Omega_f\}$. Given $R\in\mathcal{R}$, we have that $p(R)=p(R_A)\in\Omega_f^{a}$ denotes the vector of $\Omega_f$-restricted peaks of agent set $A$ at $R$, and $d(R)=d(R_D)\in\Omega_f^{n-a}$ denotes the vector of $\Omega_f$-restricted dips of agent set $D$ at $R$. Let $\Omega_f^2$ denote the set of pairs formed by alternatives in the range of $f$, i.e., $\Omega_f^2 \equiv \{(x, y) \in X^2 \, : \, x, y \in \Omega_f \mbox{ and } x \neq y\}$ and $\Omega^{C^{2}}_{f}$ the set of all ordered pairs formed by contiguous alternatives in the range of $f$, i.e., $\Omega^{C^{2}}_{f}\equiv\{(x,y) \in \Omega_f^2 : x < y  \mbox{ and } (x, y) \cap \Omega_f=\emptyset\}$.\footnote{Observe that $\Omega^{C^{2}}_{f}=\emptyset$ if either $(i)$ $|\Omega_f|=1$ and, thus, $f$ is constant, or $(ii)$ $\Omega_f$ is a closed interval and, thus,
there are no contiguous alternatives.}


\noindent We now introduce the properties imposed on the social choice rule $f$. The first two properties incentivize truthful revelation of preferences. Specifically, the first property requires that no agent ever benefits by
misrepresenting her preferences.\smallskip

\noindent \textbf{Strategy-Proofness:} For each $R\in {\cal R}$, each $i\in N$, and each $R'_i\in{\cal R}_{i}$, $ f(R)\, R_i \, f(R'_i, R_{-i})$. Otherwise, $f$ is said to be \textbf{manipulable} by agent $i$ at $R$ via $R'_i$.\medskip

\noindent Our next property states that no group of agents ever benefit by jointly misrepresenting their preferences.\smallskip

\noindent \textbf{Group Strategy-Proofness:} 

For each $R\in {\cal R}$, each $S\subseteq N$, and each $R'_S\in{\cal R}^{S}$, there is $i\in S$ such that $f(R)\, R_i \, f(R'_S,R_{-S})$.\medskip


\noindent Since we are interested in fairly considering the preferences of all agents, we also need to introduce the property of anonymity: A social choice rule $f$ is \textbf{anonymous} if for each $R\in {\cal R}$, and each permutation of agents $\sigma:N\longrightarrow N$, it follows that $f(R)=f(R_{\sigma}),$ where $R_{\sigma}$ is the profile in which the preference of each agent $i\in N$ is $R_{\sigma(i)}$. In our domain, society is partitioned into two groups, so the previous definition of anonymity cannot be directly applied. We therefore introduce an alternative anonymity property where only permutations among agents of the same type of preferences are allowed.  \smallskip

\noindent \textbf{Type-Anonymity:} For each $R\in {\cal R}$, and each permutation of agents $\sigma:N\longrightarrow N$ such that $\sigma(i) \in A $ if and only if $i \in A$, it follows that $f(R)=f(R_{\sigma}).$ \medskip


\section{Preliminaries}\label{sec2}

This section presents existing characterizations of the classes of strategy-proof and anonymous rules in both the single-peaked and the single-dipped preference domains separately, as well as of the classes of strategy-proof rules in the mixed domain considered in this paper.

\subsection{Strategy-proofness and anonymity for the single-peaked preference domain}\label{subsec21}

\noindent Strategy-proof and anonymous rules in the single-peaked preference domain were characterized by \cite{moulin1980strategy}: a rule is strategy-proof and anonymous if and only if, for any given profile, the outcome of the rule corresponds to the median between the agents' peaks and $a+1$ fixed real values.\medskip

\noindent Informally, a median of an ordered set is any value of that set such that at least half of the set is less than or equal to the proposed median and at least half is greater than or equal to it. The median value is unique when the set contains an odd number of elements. Hence, given an odd number of values $b_1,b_2,\ldots,b_{k}\in\mathbb{R}_{+}$, for some $k'\in\{1,\ldots,k\}$, $\operatorname{median}\{b_1,b_2,\ldots,b_{k}\}=b_{k'}$ if and only if [$|\{b_i, i\in\{1,\ldots,k\} : b_i\leq b_{k'}\}|\geq\frac{k+1}{2}$ and $|\{b_j, j\in\{1,\ldots,k\} : b_j\geq b_{k'}\}|\geq\frac{k+1}{2}$].\medskip

\noindent The following proposition introduces the result in \cite{moulin1980strategy}.

\begin{proposition}[Proposition 2 in \cite{moulin1980strategy}]\label{moulin}
A rule $f: {\cal R} \rightarrow X$ is strategy-proof and anonymous if and only if there exist $(a+1)$ real numbers $\gamma_1,\ldots,\gamma_{a+1}\in\mathbb{R}_{+}\cup\{+\infty,-\infty\}$ such that for each $R_A\in{\cal R}^{A}$,$$f(R_A)= \operatorname{median}\{(p(R_1),\ldots,p(R_{a}), \gamma_1,\ldots,\gamma_{a+1}\}.$$
\end{proposition} 

%
%
%

\subsection{Strategy-proofness and anonymity for the single-dipped preference domain}\label{subsec22}

\noindent \cite{manjunath2014efficient} showed that any strategy-proof rule on the single-dipped preference domain has a range of at most two alternatives (see his Lemma 7). The case of two alternatives (or the binary choice problem) has been previously studied with respect to, among other properties, the strategy-proof and anonymous rules.\medskip

\noindent For the binary choice problem, \cite{moulin1983strategy} characterized the strategy-proof and anonymous rules as ``quota majority methods". Let $X=\{x,y\}$, a profile $P\in{\cal R}$, and a value $k\in\mathbb{N}$. The corresponding rule $S_{k}(P)$ is a \textit{quota majority rule} if $S_{k}(P)$ selects $x$ if the number of votes for $x$ is at least $k$, and $y$ otherwise (it is $y$ if the number of votes for $y$ is at least $n +1- k$).\medskip

\noindent The next proposition introduces the result in \cite{moulin1983strategy}.\medskip


\begin{proposition}[Corollary page 63 in \cite{moulin1983strategy}]\label{quota}
Given $X=\{x,y\}$ and a profile $P\in{\cal R}$, a rule $f: {\cal R} \rightarrow X$ is strategy-proof and anonymous if and only if it is a quota majority method.
\end{proposition}

%
%
%

\subsection{Strategy-proofness for our domain of single-peaked and single-dipped preferences}\label{subsec23}

\noindent \cite{alcalde2024strategy} characterized the class of strategy-proof rules on the domain of single-peaked and single-dipped preferences considered in this paper. We summarize in this subsection their most relevant findings.\medskip

\noindent They found that any strategy-proof rule on this domain can be described by means of a two-step procedure as follows: in the first step, agents with single-peaked preferences report their peaks and, either a single alternative or a pair of contiguous alternatives is selected. If a single alternative is selected, then that is the final outcome. Otherwise, in the second step, each agent reports her preference over the two pre-selected alternatives, and based on that information, one alternative is chosen.\footnote{\cite{alcalde2023structure} showed that a two-step procedure applies to any domain of single-peaked and singe-dipped preferences.}\medskip

\noindent In addition to the previous result, Lemma 2 in \cite{alcalde2024strategy} shows that the outcome of a strategy-proof rule only depends on the preferences over the alternatives in the range $\Omega_f$. Therefore, $\Omega_f$ is predefined by the rule and, w.l.o.g., agents' peaks/dips coincide with their $\Omega_f$-restricted peaks/dips, i.e., for each $R\in{\cal R}$, $\rho(R)=p(R)$ and $\delta(R)=d(R)$. For convenience, we consider this condition throughout the paper. \medskip

\noindent Their result is summarized in the following corollary:
\smallskip

\begin{corollary}
\label{structure0}
Let $f: {\cal R} \rightarrow \Omega_f$ be strategy-proof. Then, there is a function $h : \Omega_f^{a} \rightarrow \Omega_f \cup \Omega^{C^{2}}_{f}$ and a set of binary decision functions $\{g_{h(\emph{\textbf p})} : {\cal R} \rightarrow h(\emph{\textbf p})\}_{h(\emph{\textbf p}) \in \Omega^{C^{2}}	_{f}}$ such that for each $\emph{\textbf p} \in \Omega_f^{a}$ and each $R \in {\cal R}$ with $p(R)= \emph{\textbf p}$, $$f(R) = \left\{
	\begin{array}{ll}
	h({\bf p})  & \mbox{ if } h({\bf p}) \in \Omega_f \\*[5pt]
	
	g_{h({\bf p})}(R) & \mbox{ if } h({\bf p}) \in \Omega^{C^{2}}_{f}.
	\end{array}
	\right.$$
\end{corollary}\medskip

\noindent Finally, we introduce the order $\leq^*$ on the set $\Omega_f\cup \Omega^{C^{2}}_{f}$ as follows: on the set $\Omega_f$, $\leq^*$ corresponds to the partial ordering $\leq$ defined on $\mathbb{R}$, while each pair in $\Omega^{C^{2}}_{f}$ is ``ordered" in between the single alternatives that form the pair. \noindent Formally, 

\begin{itemize}
\item[(i)] for each $x, y \in \Omega_f$, $x \leq^* y$ if and only if $x \leq y$,
\item[(ii)] for each $(x,y)\in \Omega^{C^{2}}_{f}$, $x <^* (x,y) <^* y$, and
\item[(iii)] for each $(x,y),(z,u)\in \Omega^{C^{2}}_{f}$, $ (x,y) =^*(z,u)$ if and only if $x=^{*}z$ and $y=^{*}u$.
\end{itemize}

\noindent To illustrate the above definition, consider, for instance, $\Omega_f=\{1,2,3\}$. Then, $\Omega^{C^{2}}_{f}=\{(1,2),(2,3)\}$. Since under order $\leq$,                                                                                                                                                                                                                                 $1<2<3$, we have that $1<^{*}2<^{*}3$, and each pair of contiguous alternatives is ordered in between the single alternatives, i.e., $1<^{*}(1,2)<^{*}2$ and $2<^{*}(2,3)<^{*}3$. As a result, the order $\leq^*$ on $\Omega_f\cup \Omega^{C^{2}}_{f}=\{1,(1,2),2, (2,3), 3\}$ is $$1<^{*}(1,2)<^{*}2<^{*}(2,3)<^{*}3.$$

\section{Strategy-proof and type-anonymous rules}\label{sec3}

\noindent In this section, we present a characterization of the strategy-proof and type-anonymous rules in our domain. Since any rule $f$ with $|\Omega_f|=1$ is both strategy-proof and type-anonymous, we focus on rules with $|\Omega_f|>1$ throughout the paper.\medskip

\noindent By \cite{alcalde2024strategy}, and as stated in Subsection~\ref{subsec23} (see Corollary~1), any strategy-proof rule can be decomposed into two steps. Hence, we analyze the structure of each step separately.
\subsection{First step}

\noindent Recall that in the first step, only agents with single-peaked preferences are considered. As aforementioned, the strategy-proof and anonymous rules on the single-peaked preferences domain were characterized by \cite{moulin1980strategy} as the median function between the peaks and a fixed collection of locations (see Subsection~\ref{subsec21}). Unlike that result, in our domain, the outcome of the first step can be not only a single alternative but also a pair of contiguous alternatives, so his result cannot be directly applied. However, we will see that a similar result holds, as the following example illustrates.\medskip

\begin{example}\label{ex1}  Let $A=\{i_1,i_2,i_3\}$ and $X=\mathbb{R}$. Consider a rule $f$ with $\Omega_f=\{1\}\cup[2,3]\cup \{4\}$ and the following collection of $(a+1)$ fixed locations: $(1,2,(3,4),(3,4))$.\medskip

\noindent Let us consider the following profiles:

\begin{itemize}
\item Let $R\in{\cal R}$ be such that $p(R)=(1,2,2)$. Since $\operatorname{median}\{1,2,2,1,2,(3,4),(3,4)\}=\operatorname{median}\{1,1,2,2,2,(3,4),(3,4)\}=2$, we have that $2\in\Omega_f$ and then, $f(R)=2$.
\item Let $R'\in{\cal R}$ be such that $p(R')=(1,4,4)$. Since $\operatorname{median}\{1,4,4,1,2,(3,4),(3,4)\}=\operatorname{median}\{1,1,2,(3,4),(3,4),4,4\}=(3,4)$, we have that $(3,4)\in\Omega^{C^{2}}_f$, and a second step is required.
\end{itemize}
\end{example}

\noindent The previous example shows the structure of the first step of any strategy-proof and type-anonymous rule in our domain. Like \cite{moulin1980strategy}, a median of the peaks of the agents and $(a+1)$ fixed locations is computed. However, the feasible collection of fixed locations can take not only real values but also values in $\Omega^{C^{2}}_f$. These fixed locations play a crucial role in this first step, and we now explain the conditions strategy-proofness and type-anonymity impose on them.\smallskip

\noindent First, note that agents of $A$ can only report single alternatives in $\Omega_f$ as peaks. Then, to select a pair as the outcome of the first step, the fixed locations must include elements in $\Omega_f^{C^{2}}$. Specifically, only those pairs reported as fixed locations can be selected as the outcome of the first step. For instance, in Example~\ref{ex1} where $\Omega^{C^{2}}_f=\{(1,2), (3,4)\}$, the pair $(1,2)$ will never be the outcome of the first step. Additionally, since all alternatives in $\Omega_f$ must be the outcome of $f$ for some profiles, we require that at least one fixed location equals either $\min\Omega_f$ or the pair in $\Omega^{C^{2}}_f$ containing $\min\Omega_f$, denoted by $\bm{\min \Omega^{C^{2}}_f}$. To see this, consider Example~\ref{ex1} with fixed locations $(2,2,(3,4),(3,4))$. It is clear that, for any vector of $\Omega_f$-restricted peaks reported, neither alternative $1$ nor pair $(1,2)$ will be the median. Consequently, alternative $1$ would never be the outcome of the rule, which is not possible given that $1\in \Omega_f$. Similarly, we impose that at least one fixed location equals either $\max\Omega_f$ or the pair in $\Omega^{C^{2}}_f$ containing $\max\Omega_f$, denoted by $\bm{\max \Omega^{C^{2}}_f}$. As a consequence of these last two constraints, any interior alternative of $\Omega_f$, i.e., $\bm{int}(\Omega_f)\equiv\{\alpha\in \Omega_f\setminus\{\min\Omega_f,\max\Omega_f\}\}$, can be chosen as the outcome for the median for some vectors of $\Omega_f$-restricted peaks. \medskip 

%

\noindent Once the collection of fixed locations is defined, we can define the first-step function $\med$ as the median between the $a$ peaks and $(a+1)$ fixed locations. We refer to function $\med$ as a \textbf{mixed median function}.  

\begin{definition}\label{mixedmedian}
\emph{Let a rule} $f: {\cal R} \rightarrow \Omega_f$. \emph{The function} $\med : \Omega_f^a \rightarrow \Omega_f \cup \Omega_f^{C^{2}}$ \emph{is a} mixed median function \emph{if there exist} $(a+1)$ \emph{fixed locations} $\gamma^{1}_{f},\dots,\gamma^{a+1}_{f} \in \Omega_f \cup \Omega_f^{C^{2}}$ with:
\begin{enumerate}
   \item[$(i)$] $\gamma_f^1 \leq^* \dots \leq^* \gamma_f^{a+1}$,
    \item[$(ii)$]  $\gamma_f^1 = \min \Omega_f$ or $\min \Omega_f^{C^2}$,
    \item[$(iii)$] $\gamma_f^{a+1} = \max \Omega_f$ or $\max \Omega_f^{C^2}$,
 
\end{enumerate}

\noindent such that for all $R \in \mathcal{R}$,
\[
    \med(p(R)) = \operatorname{median}\{p(R), \gamma_f^1, \dots, \gamma_f^{a+1}\}.
\]
\end{definition}

\noindent Observe that if the median coincides with an $\Omega_f$-restricted peak, the outcome of $\med$ is a single alternative. A pair appears as the outcome of the first step only when the median coincides with a fixed location that takes a value in $\Omega_f^{C^{2}}$. We denote the collection of fixed locations by $\mathbf{\gamma}_{f}=\{\gamma_{f}^{1},\ldots,\gamma_{f}^{a+1}\}$. Moreover, once we know $\Omega_f$ and $\mathbf{\gamma}$, we can define the range of the mixed median function, i.e., the elements that appear as the outcome of $med$. Formally, $\Omega_{\med}\equiv\{\alpha\in \Omega_f\cap\Omega^{C^{2}}_f:\exists R\in {\cal R}$ such that $\mathrm{med}(p(R))=\alpha\}$. In particular, $\Omega_{\med}=int{(\Omega_f)}\cup\{\gamma_f^1, \ldots,\gamma_f^{a+1}\}$. Going back to Example~\ref{ex1}, we have $\Omega_{\med}=\{1\}\cup[2,3]\cup\{(3,4)\}$. Observe that $4$ is not included in $\Omega_{\med}$ because there is no fixed location at $4$. However, $4$ belongs to pair $(3,4)\in \Omega_{med}$ and it will be the outcome of $f$ for some profile.\medskip

\subsection{Second step}

\noindent When the outcome of the first step is a pair of contiguous alternatives, the final outcome is determined in the second step by choosing one of the preselected alternatives. Recall that \cite{moulin1984generalized} characterized the strategy-proof and anonymous rules as ``quota-majority methods". In our domain, the type of preference of the agents plays a crucial role and it is indeed the main difference with \cite{moulin1984generalized}'s result, as illustrated in the following example. 

\begin{example}[Cont. of Example \ref{ex1}]\label{ex2} Let $D=\{j_1,j_2,j_3\}$, and consider the following quotas for $(3,4)$ where the first number of each pair corresponds to the quota for the agents in $A$ and the second one, for the agents in $D$: $\{q_{(3,4)}\}=\{(0,2),(1,1)\}$.

\noindent Let us consider the following profiles:
\begin{itemize}
%
\item Let $R'_D\in{\cal R}$ be such that $d(R_{D}')=d(R')= (2,2,1)$. Recall that $p(R')=(1,4,4)$ and then, $\med(p(R'))=(3,4)$. Since $|\{i\in A: 3\, P'_i\, 4\}|=1$ and $|\{j\in D: 3\, P'_j\, 4\}|=0$, we have that $(1,0)\not\geqq (0,2)$ and $(1,0)\not\geqq (1,1)$.\footnote{$(x,y)\geqq (z,u)$ if and only if $x \geq z$ and $y\geq u$. Similarly, $(x,y)\equiv (z,u)$ if and only if $x = z$ and $y = u$.} Hence, $f(R')=4$. 

\item Let $\bar{R}'_D\in{\cal R}$ be such that $d(\bar{R}')= (2,2,4)$. Recall that $p(R')=(1,4,4)$ and then, $\med(p(R'_A,\bar{R}'_D))=(3,4)$. Since $|\{i\in A: 3\, P'_i\, 4\}|=1$ and $|\{j\in D: 3\, \bar{P}'_j\, 4\}|=1$, we have that $(1,1)\equiv(1,1)$. Hence, $f(R'_A, \bar{R}'_D)=3$.
\end{itemize}

\end{example}

\noindent This example shows the structure of the second step of any strategy-proof and type-anonymous rule on our domain. Like in \cite{moulin1984generalized}, quotas have to be defined. However, instead of a single quota, we establish a pair of quotas $q=(q^{A},q^{D})$, where $q^{D}$, called $A$-quota, represents the quota for the agents with single-peaked preferences and $q^{D}$, called $D$-quota, represents the quota for the agents with single-dipped preferences. We refer to the pair $q$ as a \textbf{double-quota}. It is relevant to mention that multiple double-quotas may arise for each pair $\Omega_{\med}\cap\Omega_f^{C^{2}}$, and that they may differ from pair to pair. Therefore, for each pair $(x,y)\in \Omega_{\med}\cap\Omega_f^{C^{2}}$, we define a collection of double-quotas $\{q_{(x,y)}=(q^A_{(x,y)},q^D_{(x,y)})\}$.\footnote{If there is no agents with single-peaked preferences, i.e., $A=\emptyset$, then for each $(x,y)\in \Omega_{\med}\cap\Omega_f^{C^{2}}$, we assume $q^A_{(x,y)}=0$ in each double-quota.}\medskip

\noindent As was the case for the fixed collection of locations in the first step, the role of these double-quotas is fundamental. In what follows, we explain the restrictions that strategy-proofness and type-anonymity impose on them. These conditions are formally introduced in Definition \ref{doublequota}. Condition $(i)$ requires that at least one agent with single-dipped preferences agrees on the decision: that is, $q^{D}$ must be strictly positive. This condition is satisfied in Example~\ref{ex2} and the reason behind it is the fact that all agents with single-peaked preferences have already ``voted" in the first step. For convenience, condition $(ii)$ states that only minimal double-quotas are considered.\footnote{Note that this does not mean considering a double-quota formed by the minimal $A$-quota and $D$-quota. Consider, for instance, a pair $(x,y)\in \Omega_{\med}\cup\Omega_{f}^{C^{2}}$ with double-quotas $\{(2,3),(3,1)\}.$ In this case, alternative $x$ is chosen if at least either $2$ agents from $A$ and $3$ agents from $D$ or $3$ agents from $A$ and $1$ agent from $D$ prefer $x$ to $y$. If we set the minimum between each quota, we obtain pair $(2,1)$, so $x$ is selected when at least $2$ agents from $A$ and $2$ agents from $D$ prefer $x$ to $y$, which is not correct.} Finally, $q^{A}$ is determined in the first step. Recall from Example~\ref{ex1} that $\mathbf{\gamma}_f=(1,2,(3,4),(3,4))$. Then, pair $(3,4)$ is the median only if either one agent or no agent reports a peak to the left of $(3,4)$. Hence, the only possible values for $q^{A}_{(3,4)}$ are $0$ and/or $1$. This, together with the minimal assumption in condition $(ii)$, requires the existence of at least a double-quota where the $A$-quota is minimal, i.e., $q^{A}_{(3,4)}=0$ (condition $(iii)$). We call a collection of double-quotas satisfying these conditions a \textbf{set of minimal double-quotas}.\medskip

\noindent Before introducing the formal definition, we present the following notation: Let $\alpha\in \Omega_{\med}$, we denote the number of fixed locations at $\alpha$ by $M^{\alpha}$ , i.e., $M^{\alpha}\equiv|\{ \gamma_{f}^{j},\ j\in\{1,\ldots,a+1\}: \gamma_{f}^{j}=^{*}\alpha\}|$. Similarly, we denote the number of fixed locations to the left of or at $\alpha$ by $\underline{M}^{\alpha}$, i.e., $\underline{M}^{\alpha}\equiv|\{ \gamma_{f}^{j},\ j\in\{1,\ldots,a+1\}: \gamma_{f}^{j}\leq^{*}\alpha\}|$.\medskip

\begin{definition}
\label{doublequota}
 \emph{Let a rule} $f: \mathcal {R}\rightarrow \Omega_f$ \emph{and a mixed median function} $\med$, \emph{with} $\Omega_{\med}$. \emph{For each} $(x,y) \in \Omega_{\med} \cap \Omega^{C^{2}}_{f}$, \emph{a} set of minimal double-quotas $\{q_{(x,y)}\}_{l=1}^{\bar{l}}$ \emph{with} $\bar{l}\leq M^{(x,y)}$ \emph{is such that:}
 \begin{itemize}
 \item[(i)] $q^{A}_{(x,y)}\in\mathbb{N}\cup\{0\}$, and $q^{D}_{(x,y)}\in\mathbb{N}$,
 \item[(ii)] \emph{for each} $l,l'\in\{1,\ldots,\bar{l}\}$ \emph{with} $l\neq l'$, \emph{there are no} $q_{(x,y)}^{l}, q_{(x,y)}^{l'}$ \emph{such that} $q_{(x,y)}^{l}\geqq q_{(x,y)}^{l'}$ \emph{or} $q_{(x,y)}^{l'}\geqq q_{(x,y)}^{l}$, \emph{and}
 
\item[(iii)] \emph{there is} $l\in\{1,\ldots,\bar{l}\}$ \emph{and a double-quota} $q_{(x,y)}^{l}$ \emph{such that} $q^{A,l}_{(x,y)}=(a+1)-\underline{M}^{(x,y)}$.
\end{itemize}
\end{definition}

\noindent Note that, by the definition of median, $(a+1)-\underline{M}^{(x,y)}$ corresponds to the minimal number of $\Omega_f$-restricted peaks to the left of or at $(x,y)$ required to implement $(x,y)$ as the outcome of the median, i.e., the minimal $A$-quota. For instance, in Example~\ref{ex1}, the minimal $A$-quota for $(3,4)$ is $4-4=0$.\medskip


\noindent Given $(x,y)\in \Omega_{\med} \cap \Omega^{C^{2}}_{f}$ and the associated set of minimal double-quotas, a strategy-proof and type-anonymous rule selects $x$ if there exists a double-quota such that both the number of agents in $A$ and the number of agents in $D$ who prefer $x$ to $y$ are equal to or higher than the corresponding $A$-quota and $D$-quota. Otherwise, $y$ is chosen. We refer to this procedure as a \textbf{double-quota majority method}.\medskip

\noindent To introduce the formal definition, we first define the following notation. Let ${\bf p} \in \Omega_f^a$ and $\med(\mathbf{p}) \in \Omega^{C^{2}}_{f}$, we denote by $\underline{\med}(\mathbf{p})$ and $\overline{\med}(\mathbf{p})$ the left and the right alternative of the pair respectively. For each  $R\in\mathcal{R}$, $L^{A}_{(x,y)}(R)$ denotes the set of agents of $A$ that prefer $x$ to $y$ at $R$, i.e., $L^{A}_{(x,y)}(R)\equiv\{i\in A: x\, P_i\, y\}$. Similarly, $L^{D}_{(x,y)}(R)$ denotes the set of agents of $D$ who prefer $x$ to $y$ at $R$, i.e., $L^{D}_{(x,y)}(R)\equiv\{i\in D: x\, P_i\, y\}$.
	
\begin{definition}
\label{doublequotamethod}
 \emph{Let a rule} $f: \mathcal {R}\rightarrow \Omega_f$ \emph{and a mixed median function} $\med$ \emph{with} $\Omega_{\med} \subseteq \Omega_f\cup \Omega^{C^{2}}_{f}$. \emph{For each} ${\bf p} \in \Omega_f^a$ \emph{such that} $\med(\mathbf{p}) \in \Omega^{C^{2}}_{f}$, \emph{the associated binary function} $t_{\med({\bf p})}$ \emph{is a} double-quota majority method \emph{if there exists a set of minimal double-quotas} $\{q_{\med(\mathbf{p})}^{l}\}_{l=1}^{\bar{l}}$ with $\bar{l}\leq M^{\med(\mathbf{p})}$ \emph{such that for each} $R \in {\cal R}$ \emph{with} $p(R) = {\bf p}$,
	$$t_{\med({\bf p})}(R) = \left\{
	\begin{array}{ll}
	\underline{\med}(\mathbf{p}) & \mbox{ \emph{if} }  (|L^{A}_{\med(\mathbf{p})}(R)|, |L^{D}_{\med(\mathbf{p})}(R)|)\geqq q^{l}_{\med(\mathbf{p})} \mbox{ \emph{for some} } l\in \{1,\ldots, \bar{l}\}, \\*[5pt]
	\overline{\med}(\mathbf{p}) & \mbox{ \emph{otherwise}. } 
	\end{array}
	\right. \vspace{0.2cm}$$
\end{definition}

\subsection{The characterization}

\noindent The next theorem characterizes the class of strategy-proof and type-anonymous rules on our domain. It states that a  rule $f$ is strategy-proof and type-anonymous if and only if it can be decomposed into a mixed median function $\med$ and a collection of double-quota majority methods $\{t_{\med(\mathbf{p})}\}_{\med(\mathbf{p})\in \Omega_{\med}\cap\Omega^{C^{2}}_{f}}$.

\begin{theorem}
\label{theorem1}
The following statements are equivalent:
\begin{itemize}
\item[(i)] $f:\mathcal{R}\rightarrow\Omega_f$ is strategy-proof and type-anonymous.
\item[(ii)] $f:\mathcal{R}\rightarrow\Omega_f$ is group strategy-proof and type-anonymous.
\item[(iii)] There is a mixed median function $\med$ and a set of double-quota majority methods $\{t_{\med(\mathbf{p})} : {\cal R} \rightarrow \med(\mathbf{p})\}_{\med(\mathbf{p})\in \Omega_{\med}\cap\Omega^{C^{2}}_{f}}$ such that for each $R\in {\cal R}$ with $p(R)=\mathbf{p}$, $$f(R) = \left\{
	\begin{array}{ll}
	\med({\bf p})  & \mbox{ if } \med({\bf p}) \in \Omega_f \\*[5pt]
	
	t_{\med({\bf p})}(R) & \mbox{ if } \med({\bf p}) \in \Omega^{C^{2}}_{f}.
	\end{array}
	\right.$$
\end{itemize}
\end{theorem}

\noindent This result generalizes the results in \cite{moulin1980strategy,moulin1983strategy} for the single-peaked and single-dipped preference domains respectively. If $D=\emptyset$, the fixed locations only take values in $\Omega_f$ and the outcome of the median gives the final outcome of the rule. Otherwise, if $A=\emptyset$, $|\Omega_f|=2$ and a binary decision problem is faced. In addition, since we only have agents of $D$, each $A-$ quota is zero, and the final outcome is chosen by a simple quota majority method.\medskip

\noindent Observe also that strategy-proofness is equivalent to group strategy-proofness in our domain. This follows from Theorem 2 in \cite{barbera2010individual}, which shows that this equivalence exists if the domain satisfies the condition of indirect sequential inclusion. Example $(viii)$ of Section 4.5 in \cite{barbera2010individual} mentions that our domain satisfies indirect sequential inclusion.\medskip

\noindent Finally, the proof of this theorem is done in Section~\ref{sec5}.

\section{An alternative characterization}\label{sec4}

\noindent In this section, we provide an alternative characterization of the strategy-proof and type-anonymous rules based on the family characterized by \cite{alcalde2024strategy}. \medskip

\subsection{First step}

\noindent Recall that the first step consists of taking the median of the peaks and a fixed collection of locations. These fixed locations are exogenous, and once they are defined, we know both the range of the mixed median function and how many agents need to declare a peak at the left of or at a specific alternative to implement that alternative as the median. This can be understood as the ``support" required to implement an alternative. Hence, for each element in $\Omega_{\med}$, we can define a set of coalitions that represent the ``support" required to implement that element. We then define a function $\omega$ that, starting form the left under order $\leq^{*}$, chooses the first element in its range that receives ``enough" support. Let us illustrate the process with the following example:

\begin{example}\label{ex3} Let $A=\{i_1,i_2,i_3\}$, and $X=\mathbb{R}$. Consider a rule $f$ with $\Omega_f=\{1\}\cup[2,3]\cup \{4\}$ and the following fixed collection of locations: $\mathbf{\gamma}_f =(1,2,(3,4),(3,4))$. Hence,
\begin{itemize}
\item $\Omega_{\med}=\{1\}\cup[2,3]\cup \{(3,4)\}$;
\item Alternative $1$ is implemented if $3$ agents report $1$ as the peak.
\item Alternative $x\in[2,3]$ is implemented if $2$ agents report any alternative to the left of $x$ or $x$ itself as the peak.
\item Pair $(3,4)$ is implemented if no agent or $1$ agent reports any alternative to the left of $(3,4)$ as the peak.
\end{itemize}

Thus, we can define for each element in $\Omega_{\med}$, a set of ``winning" coalitions based on the minimal support required to implement the element as follows:

\begin{itemize}
\item $\mathcal{L}(1)=\{i_1,i_2,i_3\}=\{S\subseteq A: |S|=3\}$;
\item For each $x\in[2,3]$, $\mathcal{L}(x)=\{S\subseteq A: |S|\geq 2\}$;
\item $\mathcal{L}(3,4)=2^A$.
\end{itemize}

\noindent Let us consider the following profiles:

\begin{itemize}
\item Let $R\in{\cal R}$ be such that $p(R)=(1,2,2)$. Then, $\{i \in A: p(\Omega_i) \leq^* 2\}=\{i_1,i_2,i_3\} \in {\cal L}(2)$ and $\{i \in A: p(\Omega_i) \leq^* 1\}=\{i_1\} \notin {\cal L}(1)$. Hence, $\omega(p(R))=2\in\Omega_f$ and then, $f(R)=2$.
\item Let $R'\in{\cal R}$ be such that $p(R')=(1,4,4)$. Then, $\{i \in A: p(R'_i) \leq^* (3,4)\}=\{i_1\} \in {\cal L}(3,4)$ and for each $\alpha\leq^{*}(3,4)$, $\{i \in A: p(R'_i) \leq^* \alpha\}=\{i_1\} \notin {\cal L}(\alpha)$. Hence, $\omega(p(R'))=(3,4)\in\Omega^{C^{2}}_f$ and a second step is required.
\end{itemize}
\end{example}

\noindent This example illustrates the structure of any strategy-proof and type-anonymous rule on our domain. There are three main aspects to consider: first, the range of the first step. In Section \ref{sec3}, $\Omega_{\med}$ is derived from $\Omega_f$ and $\mathbf{\gamma}_f$. However, in this second result, we assume that the range of $\omega$, $\Omega_{\omega}$, is predefined by the rule given that, as shown in the example, the winning coalitions are defined only for the alternatives in the range of the first step. The range $\Omega_{\omega}$ must satisfy three conditions: $(i)$ either $\min\Omega_f$ or $\min\Omega^{C^{2}}_f$ belongs to $\Omega_{\omega}$; $(ii)$ either $\max\Omega_f$ or $\max\Omega^{C^{2}}_f$ belongs to $\Omega_{\omega}$; and $(iii)$ any interior alternative of $\Omega_f$ belongs to $\Omega_{\omega}$, i.e., $\mathit{int}(\Omega_f) \subseteq \Omega_{\omega} \subseteq \Omega_f\cup \Omega^{C^{2}}_{f}$. Note that conditions $(i)$ and $(ii)$ are equivalent to conditions $(i)$ and $(ii)$ in Definition~\ref{mixedmedian}, while condition $(iii)$ is equivalent to $\Omega_{\med}=int{(\Omega_f)}\cup\{\gamma_f^1, \ldots,\gamma_f^{a+1}\}$. This leads to Proposition 2 in \cite{alcalde2024strategy}. Formally, \medskip

\begin{proposition}[Proposition 2 in \citet{alcalde2024strategy}]\label{range}

Let $f: {\cal R} \rightarrow \Omega_f$ be strategy-proof. Then, $\Omega_\omega$ is such that 

\begin{itemize}
\item[(i)] if $\min\Omega_f^{C^{2}}\notin \Omega_\omega$, then $\min\Omega_f\in \Omega_\omega$,
\item[(ii)] if $\max\Omega_f^{C^{2}}\notin \Omega_\omega$, then $\max\Omega_f\in \Omega_\omega$, and
\item[(iii)] $\mathit{int}(\Omega_f) \subseteq \Omega_{\omega} \subseteq \Omega_f\cup \Omega^{C^{2}}_{f}$.\footnote{The order of the items here differs from the one in \cite{alcalde2024strategy}.}
\end{itemize}
\end{proposition}

\noindent Note that the extreme alternatives in $\Omega_f$, as well as some pairs of contiguous alternatives, may or may not be included in the range of $\omega$.\smallskip

\noindent Once $\Omega_{\omega}$ is defined, for each element $\alpha\in \Omega_{\omega}$, we construct a set of `winning" coalitions $\mathcal{L}(\alpha)$. These coalitions represent the minimal support an alternative needs to be implemented and must satisfy the following conditions, which will be formally introduced in Definition~\ref{leftsystem}. Condition $(i)$ states that if a coalition is ``strong enough” to support an alternative, so are its supercoalitions. Observe that this is satisfied in the previous example. Condition $(ii)$ states that if a coalition is ``strong enough” to support an alternative, it is also ``strong enough” to support any higher alternative under order $\leq^{*}$. Note that in Example~\ref{ex3}, ${\cal L}(1)\subset {\cal L}(x)\subset {\cal L}(3,4)$ for each $x\in[2,3]$. The next two conditions refer to the role of the emptyset. By condition $(i)$, if the emptyset is ``strong enough" to implement an element in $\Omega_{\omega}$, then any coalition is ``strong enough" to implement that element. Consequently, no element to its right will be implemented. To see this, consider Example~\ref{ex3} with ${\cal L}(2)=2^{A}$ and a profile $R\in{\cal R}$ such that $p(R)=(3,3,3)$. In this case, alternative $2$ is chosen in the first step since $2$ is the first element with ``enough support" given that $\{i \in A: p(\Omega_i) \leq^* 2\}=\{\emptyset\}\in {\cal L}(2)$ and $\{i \in A: p(\Omega_i) \leq^* 1\}=\{\emptyset\}\notin {\cal L}(1)$. Since ${\cal L}(2)=2^{A}$, any coalition is ``strong enough" to support any element to the left of $2$ under order $\leq^{*}$. Therefore, for any profile $R\in{\cal R}$, $\omega(p(R))\leq^{*} 2$, which is not possible given that $(2,3]\cup\{(3,4)\}\in \Omega_{\omega}$. As a result, two main conditions emerge: if $\max\Omega_f$ does not exist, the emptyset cannot be ``strong enough" to implement any element in $\Omega_{\omega}$ (condition $(iii)$); otherwise, the emptyset could only belong to ${\cal L}(\max \Omega_{\omega})$. Indeed, condition $(iv)$ states that when $\max\Omega_f$ exists but $\max\Omega_f\notin \Omega_{\omega}$, the emptyset belongs to ${\cal L}(\max \Omega_{\omega})$. Going back to Example~\ref{ex3}, we have $\max\Omega_f=4\notin \Omega_{\omega}$, and $\max \Omega_{\omega}=(3,4)$. Suppose that ${\cal L}(3,4)=2^{A}\setminus\{\emptyset\}$ and consider $R\in{\cal R}$ such that $p(R)=(4,4,4)$. Then, for any $\alpha\in \Omega_{\omega}$, $\{i \in A: p(\Omega_i) \leq^* \alpha\}=\{\emptyset\}\notin{\cal L}(\alpha)$, implying that the outcome of $\omega$ would not exist and therefore, $\omega$ would not be well-defined. Finally, condition $(v)$ says that if a coalition is ``strong enough" to implement an element in $\Omega_{\omega}$, so are all coalitions of the same size. A coalition system $\mathcal{L}$ satisfying all these conditions is called a \textbf{type-anonymous left coalition system}. Formally,

\begin{definition}
\label{leftsystem}
\emph{Let a rule} $f: {\cal R} \rightarrow \Omega_f$ \emph{with} $\mathit{int}(\Omega_f) \subseteq \Omega_{\omega} \subseteq \Omega_f\cup \Omega^{C^{2}}_{f}$. \emph{A} type-anonymous left coalition system on $\Omega_{\omega}$ \emph{is a correspondence} $\mathcal{L} : \Omega_{\omega} \rightarrow 2^A$ \emph{that assigns to each} $\alpha \in \Omega_{\omega}$ \emph{a collection of coalitions} $\mathcal{L}(\alpha)$ \emph{such that:}
\begin{enumerate}
\item[$(i)$] \emph{if} $S\in \mathcal{L}(\alpha)$ \emph{and} $S \subset S'$, \emph{then} $S' \in \mathcal{L}(\alpha)$,
\item[$(ii)$] \emph{if} $\alpha <^* \beta$ \emph{and} $S \in \mathcal{L}(\alpha)$, \emph{then} $S \in \mathcal{L}(\beta)$, 
\item[$(iii)$] \emph{if} $\max\Omega_f$ \emph{does not exist}, \emph{then for each} $\alpha\in \Omega_{\omega}$, $\emptyset\notin {\cal L}(\alpha)$,
\item[$(iv)$] \emph{if} $\max\Omega_f$ \emph{exists and} $\max \Omega_f \notin \Omega_{\omega}$, \emph{then for each} $\alpha\in \Omega_{\omega}\setminus\{\max \Omega_{\omega}\}$, $\emptyset\in {\cal L}(\max \Omega_{\omega}) \setminus{\cal L}(\alpha)$, \emph{and}
\item[$(v)$] \emph{if} $S\in \mathcal{L}(\alpha)$ \emph{and} $S'\subseteq A$ such that $|S'|=|S|$, \emph{then} $S' \in \mathcal{L}(\alpha)$.
\end{enumerate}
\end{definition}
\medskip

\noindent Conditions $(i)$-$(iv)$ derive from strategy-proofness \citep[see Definition 1 in][]{alcalde2024strategy}, while condition $(v)$ results from type-anonymity, which focuses only on the number of agents of each type of preference that support an alternative, making the identities of those agents irrelevant.\medskip

\noindent Finally, given a type-anonymous left coalition system, we define, for any vector of $\Omega_f$-restricted peaks, the first-step function $\omega$ as follows: starting form the left under order $\leq^{*}$, $\omega$ chooses the first element in $\Omega_{\omega}$ that receives ``enough" support. By ``enough" support we mean that the set of agents whose $\Omega_f$-restricted peaks are to the left of or at that element under $\leq^{*}$ must coincide with a winning coalition defined for that element. The function $\omega$ is a \textbf{generalized median voter function}.\footnote{The term ``generalized median voter function" has been previously introduced in the literature \cite[see, for instance,][]{barbera2011strategyproof}. Even though the definition does not apply a median, the name comes from the median voter functions defined in \cite{moulin1980strategy}, which are formally introduced in Subsection 2.1.} 

\begin{definition}
\label{voterscheme}
\emph{Let a rule} $f: \mathcal {R}\rightarrow \Omega_f$ \emph{and a type-anonymous left coalition system} $\mathcal{L}$ \emph{on} $\Omega_{\omega}$, with $\mathit{int}(\Omega_f) \subseteq \Omega_{\omega} \subseteq \Omega_f\cup \Omega^{C^{2}}_{f}$. \emph{The associated} generalized median voter function $\mathbf{\omega}$ \emph{is as follows:} \emph{for each} ${\bf p} \in \Omega_f^a$ \emph{and each} $R \in {\cal R}$ \emph{such that} $p(R) = {\bf p}$,
$$\omega(\mathbf{p})= \alpha \mbox{\emph{ if and only if}}$$ $$ \{i \in A: p(\Omega_i) \leq^* \alpha\} \in \mathcal{L}(\alpha)$$ \
\emph{and for each} $\beta \in \Omega_{\omega}$ \emph{such that} $\beta <^* \alpha$, 
$$\{i \in A: p(\Omega_i) \leq^* \beta\}\notin \mathcal{L}(\beta).$$
\end{definition}

\noindent There is a last aspect of the first step that is worthy mentioning. Let $x,y\in\Omega_f$ and $(x,y)\in \Omega^{C^{2}}_f$. If $\mathcal{L}(x)=\mathcal{L}(x,y)$, then pair $(x,y)$ can never be the output of the first step. This happens because the agents of $A$ only report alternatives in $\Omega_f$ as peaks and $x<^{*} (x,y)$.\medskip

\noindent The following proposition introduces the structure of the first step of any strategy-proof and type-anonymous rule on our domain. It states that if rule $f$ is strategy-proof and type-anonymous, $\omega$ is a type-anonymous generalized median voter function.

\begin{proposition}
\label{generalized}
Let $f: {\cal R} \rightarrow \Omega_f$ be strategy-proof and type-anonymous. Then, the function $\omega$ is a type-anonymous generalized median voter function on a set $\Omega_{\omega}$ satisfying the conditions of Proposition~\ref{range}.
\end{proposition}

\begin{proof}
\noindent By the proof of Proposition 3 in \cite{alcalde2024strategy}, it only remains to show that type-anonymity implies condition $(v)$ in Definition~\ref{leftsystem}. To see this, we prove that given $\alpha\in \Omega_{\omega}$, if $S\in \mathcal{L}(\alpha)$ and $S'\subset A$ is such that $|S'|=|S|$ and $S' \notin \mathcal{L}(\alpha)$, then $f$ is not type-anonymous.\smallskip

\noindent Let $\alpha\in \Omega_{\omega}$ and $S\subset A$ such that $S\in {\cal L}(\alpha)$. Suppose by contradiction that there is $S'\subset A$ such that $|S'|=|S|$ and $S'\notin {\cal L}(\alpha)$. Consider $R\in {\cal R}$ such that $p(\Omega_i)\leq^{*} \alpha \Leftrightarrow i\in S$. Since $S\in \mathcal{L}(\alpha)$, by Definition \ref{voterscheme}, $\omega(p(R))\leq^{*}\alpha$. Consider now a permutation $\sigma$ such that for each $j\in S$, $\sigma(j)\in S'$ and for each $k\in N\setminus S$, $\sigma(k)=k$. Note that, by construction, $\Omega_{\sigma}$ is such that $p((\Omega_{\sigma})_i)\leq^{*} \alpha \Leftrightarrow i\in S'$. Since $S'\notin \mathcal{L}(\alpha)$, by Definition \ref{voterscheme}, $\omega(p(\Omega_{\sigma}))>^{*}\alpha$. Hence, $\omega(p(R))\neq\omega(p(\Omega_{\sigma}))$ and there is $x\in \Omega_{\omega}$ such that, w.l.o.g., $x\in \omega(p(R))\setminus \omega(p(\Omega_{\sigma}))$. Consider $R'_D\in{\cal R}^{D}$ such that $f(\Omega_A,R'_D)=x$. Since $x\notin\omega(p(\Omega_{\sigma}))$ and $\omega(p(\Omega_{\sigma}))=\omega(p(\Omega_A,R'_D)_{\sigma})$, $x\notin\omega(p(\Omega_A,R'_D)_{\sigma})$. Therefore, $f((\Omega_A,R'_D)_{\sigma})\neq x=f(\Omega_A,R'_D)$, and $f$ is not type-anonymous.\end{proof}

\subsection{Second step}

\noindent Recall that the second step of any strategy-proof rule consists of a binary choice problem. In this step, we define for each pair in $\Omega_{\omega}$, a set of coalitions that represent the minimal ``support" required to implement the left alternative of the pair. That is, the left alternative is selected if the coalition formed by all agents who prefer the left alternative to the right one coincides with one of the coalitions defined for that pair or is a supercoalition of any of them. Let us illustrate the process with the following example:\medskip

\begin{example}[Cont. of Example \ref{ex3}]\label{ex4} Let $D=\{j_1,j_2,j_3\}$, and consider $\{q_{(3,4)}\}=\{(0,2),(1,1)\}$.

\begin{itemize}
\item $q_{(3,4)}=(0,2)\Rightarrow W(3,4)=\{\{j_1,j_2\},\{j_1,j_3\},\{j_2,j_3\}\}$
\item $q_{(3,4)}=(1,1)\Rightarrow W(3,4)=\{\{i_1,j_1\},\{i_1,j_2\},\{i_1,j_3\},\{i_2,j_1\},\{i_2,j_2\},\{i_2,j_3\},\{i_3,j_1\},\{i_3,j_2\},\{i_3,j_3\}\}$
\end{itemize}

\noindent Hence, $W(3,4)=\{S\subseteq N: |S\cap A|=0$ and $|S\cap D|=2\} \cup\{S'\subseteq N: |S'\cap A|=1$ and $|S'\cap D|=1\}$.\medskip

\noindent Let us consider the following profiles:
\begin{itemize}
\item Let $R'_D\in{\cal R}$ be such that $d(R')= (2,2,1)$. Recall that $p(R')=(1,4,4)$ and then, $\omega(p(R'))=(3,4)$. Since $\{i\in A: 3\, P'_i\, 4\}=\{i_1\}$ and $\{j\in D: 3\, P'_j\, 4\}=\{\emptyset\}$, we have that $\{i_1\}\not\supseteq S\in W(3,4)$. Hence, $f(R')=4$. 

\item Let $\bar{R}'_D\in{\cal R}$ be such that $d(R')= (2,2,4)$. Recall that $p(R')=(1,4,4)$ and then, $\omega(p(R'_A,\bar{R}'_D))=(3,4)$. Since $\{i\in A: 3\, P'_i\, 4\}=\{i_1\}$ and $\{j\in D: 3\, \bar{P}'_j\, 4\}=\{j_3\}$, we have that $\{i_1,j_3\}\in W(3,4)$. Hence, $f(R'_A, \bar{R}'_D)=3$.
\end{itemize}
\end{example}

\noindent This example shows the structure of the second step of any strategy-proof and type-anonymous rule on our domain. For each $(x,y)\in \Omega_{\omega}\cap \Omega^{C^{2}}_f$, we define a set of ``winning" coalitions that must satisfy certain restrictions. First, we consider minimal coalitions, meaning that there is no other coalition that contains them or is contained in them (condition $(i)$). Moreover, condition $(ii)$ states that each coalition contains at least an agent with single-dipped preferences. Recall that $(x,y)$ is chosen by $\omega$ only when the agents that support $(x,y)$ do not support $x$, i.e., when the agents form a coalition in ${\cal L}(x,y)\setminus {\cal L}(x)$. Hence, condition $(iii)$ requires that there must be a coalition whose agents with single-peaked preferences belongs to ${\cal L}(x,y)\setminus {\cal L}(x)$. Finally, condition $(iv)$ states that if a coalition is ``strong enough" to implement the left alternative of the pair, then so are all coalition formed by the same number of both agents with single-peaked preferences and agents with single-dipped preferences. A set of coalitions $W$ satisfying these conditions is called a \textbf{collection of type-anonymous left-decisive sets}.\medskip


\begin{definition}
\label{winningdef}
 \emph{Let a rule} $f: \mathcal {R}\rightarrow \Omega_f$  \emph{and a generalized median voter function} $\omega$ \emph{on a set} $\Omega_{\omega}$ that satisfies the conditions of Proposition~\ref{range}, with \emph{associated type-anonymous left coalition system} ${\cal L}$. \emph{A} collection of type-anonymous left-decisive sets \emph{is a correspondence} $W: \Omega_{\omega} \cap \Omega^{C^{2}}_{f}\rightarrow 2^N$ \emph{that assigns to each} $(x,y)\in \Omega_{\omega} \cap \Omega^{C^{2}}_{f}$ \emph{a collection of coalitions} $W(x,y)$ \emph{such that:}
 \begin{itemize}
 \item[(i)] \emph{for each} $S, S'\in W(x,y)$, neither $S'\subset S$ nor $S\subset S'$,
		\item[(ii)] \emph{for each} $S\in W(x,y)$, $S\cap D\neq \emptyset$,
		\item[(iii)] \emph{for each minimal coalition} $B$ of ${\cal L}(x,y)\setminus{\cal L}(x)$, \emph{there is} $S\in W(x,y)$ \emph{such that} $S\cap A=B$, and
		\item[(iv)] \emph{if} $S\in W(x,y)$, and $S'\subset N$ such that $[|S'\cap A|=|S\cap A|$ and $|S'\cap D|=|S\cap D|]$, then $S'\in W(x,y)$.
\end{itemize}
		\end{definition}
 
 \noindent Conditions $(i)-(iii)$ derive from strategy-proofness \citep[see Definition 3 in][]{alcalde2024strategy}. Condition $(iv)$ is a consequence of type-anonymity, which focuses only on the number of agents of each type of preference that support an alternative, making the identities of those agents irrelevant.\medskip
 
\noindent Finally, given a collection of type-anonymous left-decisive sets $W$, we define for each $(x,y)\in \Omega_{\omega} \cap \Omega_f^{C^{2}}$, a binary choice function $g_{(x,y)}$ as follows: alternative $x$ is chosen if the set of agents who prefer $x$ to $y$ coincides with some coalition in $W(x,y)$ or is a supercoalition of one of them. Otherwise, $y$ is chosen. These binary functions are called \textbf{voting by collections of type-anonymous left-decisive sets}.\medskip
 
\noindent Before introducing the formal definition, we need the following concepts. Let $\omega(\mathbf{p}) \in \Omega^{C^{2}}_{f}$, we denote by $\underline{\omega}(\mathbf{p})$ and $\overline{\omega}(\mathbf{p})$ the left and the right alternative of the pair respectively. For each $R \in {\cal R}$ such that $p(R)=\mathbf{p}$, let $L_{\omega(\mathbf{p})}(R)$ be the set of agents that prefer $\underline{\omega}(\mathbf{p})$ to $\overline{\omega}(\mathbf{p})$ at $R$, i.e., $L_{\omega(\mathbf{p})}(R)\equiv\{i\in N: \underline{\omega}(\mathbf{p})\, P_i\, \overline{\omega}(\mathbf{p})\}$.\medskip

\begin{definition}
\label{votingdef}
 \emph{Let a rule} $f: \mathcal {R}\rightarrow \Omega_f$ \emph{and a generalized median voter function} $\omega$ \emph{on a set} $\Omega_{\omega}$ satisfying the conditions of Proposition~\ref{range}, \emph{and with associated type-anonymous left coalition system} ${\cal L}$. \emph{For each} ${\bf p} \in \Omega_f^a$ \emph{such that} $\omega(\mathbf{p}) \in \Omega_{\omega} \cap \Omega^{C^{2}}_{f}$, \emph{the associated binary function} $g_{\omega({\bf p})}$ \emph{is a} voting by collections of type-anonymous left-decisive sets \emph{if there is collection of type-anonymous left-decisive sets} $W$ \emph{such that for each} $R \in {\cal R}$ \emph{with} $p(R) = {\bf p}$,
	$$g_{\omega({\bf p})}(R) = \left\{
	\begin{array}{ll}
	\underline{\omega}(\mathbf{p}) & \mbox{ \emph{if} } S \subseteq L_{\omega({\bf p})}(R) \mbox{ \emph{for some} } S \in W(\omega({\bf p})), \\*[5pt]
	\overline{\omega}(\mathbf{p}) & \mbox{ \emph{otherwise}. } 
	\end{array}
	\right. \vspace{0.2cm}$$
\end{definition}

\noindent The next proposition introduces the structure of the second step of any strategy-proof and type-anonymous rule on our domain. It states that if rule $f$ is strategy-proof and type-anonymous, $g_{\omega(\mathbf{p})}$ is a voting by collections of type-anonymous left-decisive sets.

\begin{proposition}
\label{second-step2}
Let $f: {\cal R} \rightarrow \Omega_f$ be strategy-proof and type-anonymous, and $\omega$ be its associated type-anonymous generalized median voter function on $\Omega_\omega$, satisfying the conditions of Proposition~\ref{range}. Then, the family of binary functions $\{g_{\omega(\mathbf{p})} :  {\cal R} \rightarrow \omega(\mathbf{p})\}_{\omega(\mathbf{p})\in \Omega_{\omega}\cap\Omega^{C^{2}}_{f}}$ is such that for each $\omega(\mathbf{p}) \in \Omega_{\omega}\cap\Omega^{C^{2}}_{f}$, $g_{\omega(\mathbf{p})}$ is a voting by collections of type-anonymous left-decisive sets.
\end{proposition}

\begin{proof}
\noindent By the proof of Proposition 4 in \cite{alcalde2024strategy}, it only remains to show that type-anonymity implies condition $(iv)$ in Definition~\ref{winningdef}. To see this, we prove that if $S\in W(\omega({\bf p}))$, $S'\subset N$ is such that [$|S'\cap A|=|S\cap A|$ and $|S'\cap D|=|S\cap D|$], and $S'\notin W(\omega({\bf p}))$, then $f$ is not type-anonymous.\smallskip

\noindent Let $R\in {\cal R}$ be such that $\omega(p(R))\in \Omega_{\omega}\cap\Omega_f^{C^{2}}$ and $S\subset N$ be such that $S\in W(\omega(p(R)))$. Suppose by contradiction that there is $S'\subset N$ such that $|S'\cap A|=|S\cap A|$ and $|S'\cap D|=|S\cap D|$, and $S'\notin W(\omega(p(R)))$. Consider now $\bar{R}\in {\cal R}$ such that $\omega(p(\bar{R}))=\omega(p(R))$ and $L_{\omega(p(R))}(\bar{R})= S$. Since $S\in W(\omega(p(R)))=W(\omega(p(\bar{R})))$, by Definition \ref{votingdef}, $g_{\omega(p(\bar{R}))}=\underline{\omega}(p(R))$ and then, $f(\bar{R})=\underline{\omega}(p(R))$. Consider a permutation $\sigma$ such that
for each $j\in S$, $\sigma(j)\in S'$ and for each $k\in N\setminus S$, $\sigma(k)=k$. Hence, by construction, $\bar{R}_{\sigma}$ is such that $\omega(p(\bar{R}_{\sigma}))=\omega(p(\bar{R}))$ and $L_{\omega(p(R))}(\bar{R}_{\sigma})=S'$. Since $S'\notin W(\omega(p(R)))=W(\omega(p(\bar{R}_{\sigma})))$, by Definition \ref{votingdef}, $g_{\omega(p(\bar{R}_{\sigma}))}=\overline{\omega}(p(R))$ and then, $f(\bar{R}_{\sigma})=\overline{\omega}(p(R))$. Therefore, $f(\bar{R}_{\sigma})=\overline{\omega}(p(R))\neq\underline{\omega}(p(R))=f(\bar{R})$ and $f$ is not type-anonymous.\end{proof}

\subsection{The characterization}\label{result}

The next result establishes that the necessary conditions derived from strategy-proofness and type-anonymity in the former propositions are also sufficient.

\begin{theorem}
\label{theorem2}
The following statements are equivalent:
\begin{itemize}
\item[(i)] $f:\mathcal{R}\rightarrow\Omega_f$ is strategy-proof and type-anonymous.
\item[(ii)] $f:\mathcal{R}\rightarrow\Omega_f$ is group strategy-proof and type-anonymous.
\item[(iii)] There is a type-anonymous left coalition system with associated generalized median voter function $\omega$ on a set $\Omega_\omega$, satisfying the conditions of Proposition~\ref{range}, and a set of voting by collections of type-anonymous left-decisive sets $\{g_{\omega(\mathbf{p})} :  {\cal R} \rightarrow \omega(\mathbf{p})\}_{\omega(\mathbf{p})\in \Omega_{\omega}\cap\Omega^{C^{2}}_{f}}$ such that for each $R\in {\cal R}$ with $p(R)=\mathbf{p}$, $$f(R) = \left\{
	\begin{array}{ll}
	\omega({\bf p})  & \mbox{ if } \omega({\bf p}) \in \Omega_f \\*[5pt]
	
	g_{\omega({\bf p})}(R) & \mbox{ if } \omega({\bf p}) \in \Omega^{C^{2}}_{f}.
	\end{array}
	\right.$$
\end{itemize}
\end{theorem}

\begin{proof}
\noindent By the proof of Theorem 1 in \cite{alcalde2024strategy}, it only remains to show that condition $(v)$ in Definition~\ref{leftsystem} and condition $(iv)$ in Definition~\ref{winningdef} are necessary and sufficient for type-anonymity. Propositions~\ref{generalized} and \ref{second-step2} show that they are necessary, so we prove here that they are also sufficient. Consider any $f$ that is decomposable as described in $(iii)$.\smallskip

\noindent Let $R\in {\cal R}$ be such that $\omega(p(R))=\alpha$. By Definition \ref{voterscheme}, $\{i \in A: p(\Omega_i) \leq^* \alpha\} \in \mathcal{L}(\alpha)$ and for each $\beta \in \Omega_{\omega} $ such that $\beta <^* \alpha$,
$\{i \in A: p(\Omega_i) \leq^* \beta\}\notin \mathcal{L}(\beta)$. Consider now a permutation $\sigma$ such that for each $j\in A$, $\sigma(j)\in A$ and for each $k\in D$, $\sigma(k)=k$. By construction, $|\{i \in A: p((\Omega_{\sigma})_{i}) \leq^* \alpha\}|=|\{i \in A: p(\Omega_i) \leq^* \alpha\}|$ and for each $\beta \in \Omega_{\omega} $ such that $\beta <^* \alpha,$
$|\{i \in A: p((\Omega_{\sigma})_{i}) \leq^* \beta\}|=|\{i \in A: p(\Omega_{i}) \leq^* \beta\}|$. Therefore, by Definition \ref{voterscheme}, $\omega(p(\Omega_{\sigma}))=\alpha.$ If $\alpha\in \Omega_f$, then $f(\Omega_{\sigma})=\alpha=f(R)$ and $f$ is type-anonymous. Assume now that $\alpha\in \Omega_f^{C^{2}}$, i.e., $\alpha=(\underline{\alpha},\overline{\alpha})$. If $f(R)=\underline{\alpha}$, then by Definition \ref{votingdef}, $g_{\alpha}(R) =
	\underline{\alpha}$ and there is $S\in W(\alpha)$ such that $S\subseteq L_{\alpha}(R)$. Note that, by construction, $|L_{\alpha}(\Omega_{\sigma})|=|L_{\alpha}(R)|$ and, by Definition \ref{winningdef}, there is $S'\in W(\alpha)$ such that $S'\subseteq L_{\alpha}(\Omega_{\sigma})$. Then, $g_{\alpha}(\Omega_{\sigma}) =
	\underline{\alpha}$, and $f(\Omega_{\sigma})=\underline{\alpha}=f(R)$. The argument is similar if $f(R)=\overline{\alpha}$ and thus omitted. Hence, $f$ is type-anonymous. \end{proof}

\subsection{A strategy-proof rule that is not type-anonymous}

\noindent We finally provide an example of a rule that is strategy-proof (and then it belongs to the rules characterized in \cite{alcalde2024strategy}), but it is not type-anonymous (and then it does not belong to our characterized rules). We are going to consider a variation of Example~\ref{ex3}.

\begin{example}\label{ex5} Let $A=\{i_1,i_2,i_3\}$, $D=\{j_1,j_2,j_3\}$, and $\Omega_f=\{1\}\cup[2,3]\cup \{4\}$. Consider a rule $f$ with $\Omega_\omega=\{1\}\cup[2,3]\cup \{(3,4)\}$ such that: 

\begin{itemize}
\item $\mathcal{L}(1)=\{i_1,i_2,i_3\}$;
\item For each $x\in[2,3]$, $\mathcal{L}(x)=\{\{i_1,i_3\},\{i_2,i_3\},\{i_1,i_2,i_3\} \}$;
\item $\mathcal{L}(3,4)=2^A$;
\item $W(3,4)=\{S\subseteq N: |S\cap A|=0$ and $|S\cap D|=2\} \cup \{i_2,j_1\}$.
\end{itemize}

\noindent Let us consider the following profiles:

\begin{itemize}
\item Let $R\in{\cal R}$ be such that $p(R)=(1,2,4)$ and $d(R)=(2,2,4)$. Then, $\{i \in A: p(\Omega_i) \leq^* 2\}=\{i_1,i_2\} \in {\cal L}(2)$ and $\{i \in A: p(\Omega_i) \leq^* 1\}=\{i_1\} \notin {\cal L}(1)$. Hence, $\omega(p(R))=2\in\Omega_f$ and then, $f(R)=2$.
\item Let $R'\in{\cal R}^{A}$ be such that $p(R')=(2,1,4)$ and $d(R')=d(R)=(2,2,4)$. Then, $\{i \in A: p(R'_i) \leq^* (3,4)\}=\{i_1,i_2\} \in {\cal L}(3,4)$ and $\{i \in A: p(R'_i) \leq^* 3\}=\{i_1,i_2\} \notin {\cal L}(3)$. Hence, $\omega(p(R'))=(3,4)\in\Omega^{C^{2}}_f$. Since $\{i \in N: 3\, P'_i\, 4\}=\{i_1,i_2,j_3\}\notin W(3,4)$. Hence, $g(R')=4$ and $f(R')=4.$ 
\item Let $R''\in{\cal R}^{A}$ be such that $p(R'')=p(R')=(2,1,4)$ and $d(R'')=(4,2,2)$. Then, $\{i \in A: p(R''_i) \leq^* (3,4)\}=\{i_1,i_2\} \in {\cal L}(3,4)$ and $\{i \in A: p(R''_i) \leq^* 3\}=\{i_1,i_2\} \notin {\cal L}(3)$. Hence, $\omega(p(R''))=(3,4)\in\Omega^{C^{2}}_f$. Since $\{i \in N: 3\, P''_i\, 4\}=\{i_1,i_2,j_i\}\supset\{i_2,j_1\}\in W(3,4)$. Hence, $g(R'')=3$ and $f(R'')=3.$ 
\end{itemize}

\noindent Note that by construction $R'=\Omega_{\sigma}$ where $\sigma$ is a permutation such that $\sigma(i_1)=i_2$ and $\sigma(i_2)=i_1$, and for each $k\in N\setminus\{i_1,i_2\}$, $\sigma(i_k)=i_k$. However, $f(R)=2\neq 4=f(R')$. Similarly, $R''=R'_{\sigma'}$ where $\sigma'$ is a permutation such that $\sigma'(j_1)=j_3$ and $\sigma'(j_3)=j_1$, and for each $k\in N\setminus\{j_1,j_3\}$, $\sigma'(i_k)=i_k$. However, $f(R')=4\neq 3=f(R'')$. Thus, $f$ is not type-anonymous, but it can be easily checked that it belongs to the family of rules characterized in \cite{alcalde2024strategy}.
\end{example}

\section{The equivalence between the characterizations}\label{sec5}

\noindent This section shows the equivalence of Theorems~\ref{theorem1} and \ref{theorem2}. In particular, we show that part $(iii)$ of both theorems is equivalent. We will use the superscripts T1 and T2 to distinguish between both approaches when clarification seems necessary.\medskip

\noindent $\Rightarrow)$ \emph{Given a rule $f$ as in Theorem~\ref{theorem1} $(iii)$, we show that it can be described as in Theorem~\ref{theorem2} $(iii)$.}\smallskip

\noindent \emph{\underline{First step}}\smallskip

\begin{itemize}
\item We first show how, given a collection of fixed locations, a type-anonymous left coalition system can be constructed.\smallskip

\noindent Let $f$ be a rule with range $\Omega_f$ and fixed locations $\gamma_ {f}^{1},\ldots, \gamma_{f}^{a+1}\in \Omega_f\cap\Omega^{C^{2}}_{f}$ as in Definition \ref{mixedmedian}. Then, $\Omega_{\med}=int({\Omega_f})\cup \{\gamma_{f}^1,\ldots,\gamma_{f}^{a+1}\}$. We define $\Omega_{\omega}=\Omega_{\med}$ and construct a type-anonymous left coalition system ${\cal L}$ as follows: for each element $\alpha\in \Omega_{\omega}$, we first find $\underline{M}^{\alpha}$, i.e., how many fixed locations are to the left of or at $\alpha$. Recall that $(a+1)-\underline{M}^{\alpha}$ is the minimum number of $\Omega_f$-restricted peaks to the left of or at $\alpha$ required to implement $\alpha$. We then include in each ${\cal L}(\alpha)$ all the coalitions formed by $(a+1)-\underline{M}^{\alpha}$ agents and all its supercoalitions. Finally, if $\max \Omega_{\omega}\in\Omega_{f}^{C^{2}}$, then we include the emptyset in ${\cal L}(\max \Omega_{\omega})$. Formally, for each $\alpha\in \Omega_{\omega}\setminus\max \Omega_{\omega}$, ${\cal L}(\alpha)\equiv\{S\subseteq A: |S|\geq(a+1)-\underline{M}^{\alpha}\}$, and ${\cal L}(\max \Omega_{\omega})=2^{A}$ if $\max \Omega_{\omega}\in\Omega_{f}^{C^{2}}$ or ${\cal L}(\max \Omega_{\omega})=2^{A}\setminus\{\emptyset\}$ if $\max \Omega_{\omega}\in \Omega_f$. It can be easily checked that, by construction, ${\cal L}$ satisfies all conditions in Definition \ref{leftsystem}, and hence it is a type-anonymous left coalition system.

\item We now show that, given a profile $R\in\mathcal{R}$, the outcome of $\med(p(R))$ for the given collection of fixed locations coincides with the outcome of $\omega(p(R))$ for the type-anonymous left coalition system constructed.\smallskip

\noindent Consider $R\in{\cal R}$ such that $\med(p(R))=\alpha$. By the definition of the median, we know that $\underline{M}^{\alpha}+|\{i\in A: p(\Omega_i)\leq^{*}\alpha\}|\geq a+1$ and for each $\beta \in \Omega_{\omega}$ with $\beta<^{*}\alpha$, $\underline{M}^{\beta}+|\{i\in A: p(\Omega_i)\leq^{*}\beta\}|< a+1$. Therefore, 
\begin{equation}\label{eq1}
|\{i\in A: p(\Omega_i)\leq^{*}\alpha\}|\geq a+1-\underline{M}^{\alpha}
\end{equation}
and 
\begin{equation} \label{eq2}
|\{i\in A: p(\Omega_i)\leq^{*}\beta\}|< a+1-\underline{M}^{\beta}.
\end{equation}

\noindent By construction, ${\cal L}(\alpha)=\{S\subseteq A: |S| \geq(a+1)-\underline{M}^{\alpha}\}$ and ${\cal L}(\beta)=\{S\subseteq A: |S| \geq(a+1)-\underline{M}^{\beta}\}$. Hence, inequalities (\ref{eq1}) and (\ref{eq2}) imply that  

\begin{equation*}
\{i\in A: p(\Omega_i)\leq^{*}\alpha\}\in{\cal L}(\alpha)
\end{equation*}
and that for each $\beta \in \Omega_{\omega}$ with $\beta<^{*}\alpha$,
\begin{equation*} 
\{i\in A: p(\Omega_i)\leq^{*}\beta\}\notin{\cal L}(\beta).
\end{equation*}

\noindent Thus, by Definition \ref{voterscheme}, $\omega(p(R))=\alpha$.
\end{itemize}

\noindent If $\alpha\in \Omega_f$, then $f^{T2}(R)=\omega(p(R))=\alpha=\med(p(R))=f^{T1}(R)$. Otherwise, we go to the second step.\medskip

\noindent \emph{\underline{Second step}}\smallskip

\begin{itemize}
\item We first show how, given a set of double-quotas, a collection of type-anonymous left-decisive sets can be constructed.\smallskip

\noindent Let $\mathbf{p}\in\Omega_f^{a}$ with $\omega(\mathbf{p})\in \Omega_{\omega}\cap\Omega_{f}^{C^{2}}$ and a set of minimal double-quotas $\{q_{\omega(\mathbf{p})}\}_{l=1}^{\bar{l}}$ as in Definition \ref{doublequota}. Then, a collection of type-anonymous left-decisive sets $W$ is constructed as follows: For each $\omega(\mathbf{p})\in \Omega_{\omega}\cap\Omega_{f}^{C^{2}}$, $W(\omega(\mathbf{p}))$ includes all coalitions formed by the number of agents of the double-quotas defined for $\omega(\mathbf{p})$. Formally, $W(\omega(\mathbf{p}))=\{S\subseteq N: (|S\cap A|,|S\cap D|)\equiv(q^{A}_{\omega(\mathbf{p})},q^{D}_{\omega(\mathbf{p})})^{l}$ for each $l\in\{1,\ldots,\bar{l}\}\}$. It can be easily checked that, by construction, $W$ satisfies all conditions in Definition \ref{winningdef}, and hence it is a collection of type-anonymous left-decisive sets.

\item We now show that, given a profile $R\in\mathcal{R}$, the outcome of $t_{\med(p(R))}$ for the given set of double-quotas coincides with the outcome of $g_{\omega(p(R))}$ for the collection of type-anonymous left-decisive sets constructed.\smallskip

\noindent Let $R\in\mathcal{R}$ such that $\med(p(R))\in \Omega_{f}^{C^{2}}$, and assume w.l.o.g. that $f^{T1}(R)=\underline{\med}(p(R))$. Then $t_{\med(p(R))}=\underline{\med}(p(R))$ and by Definition \ref{doublequotamethod}, there is a double quota $q_{\med(p(R))}=(q^{A}_{\med(p(R))},q^{D}_{\med(p(R))})$ such that $(|L_{\med(p(R))}^{A}(R)|, |L_{\med(p(R))}^{D}(R)|)\geqq (q^{A}_{\med(p(R))},q^{D}_{\med(p(R))})$. By construction, for each $S\subseteq N$ such that [$|S\cap A|=q^{A}_{\med(p(R))}$ and $|S\cap D|=q^{D}_{\med(p(R))}$], $S\in W(\omega(p(R))).$ Then, for some $S\in W(\omega(p(R))$, we have $L_{\omega(p(R))}(R)\supseteq S$. Hence, by Definition \ref{votingdef}, $g(R)=\underline{\omega}(p(R))$ and $f^{T2}(R)=\underline{\omega}(p(R))=\underline{\med}(p(R))=f^{T1}(R)$.
\end{itemize}

\noindent $\Leftarrow)$ \emph{Given a rule $f$ as in Theorem 2 $(iii)$, we show that it can be described as in Theorem 1 $(iii)$.}\smallskip

\noindent \emph{\underline{First step}}\smallskip
\begin{itemize}
\item We first show how, given a type-anonymous left coalition system, a collection of fixed locations can be defined.\smallskip

\noindent Let $f$ be a rule, with range $\Omega_f$, and its associated type-anonymous generalized median voter function $\omega$, with range $\Omega_{\omega}$ satisfying the conditions in Proposition 3, and corresponding type-anonymous left coalition system ${\cal L}$ as in Definition \ref{leftsystem}. For each $\alpha\in \Omega_{\omega}$, let $\underline{S}_{\alpha}$ denote the smallest size of the coalitions in ${\cal L}(\alpha)$. Then, starting from the left, the fixed locations $\gamma_f^{1},\ldots,\gamma_f^{a+1}$ are located on some elements of $\Omega_{\omega}$ as follows: $M^{\min \Omega_{\omega}}=(a+1)- \underline{S}_{\min \Omega_{\omega}}$; for each $\alpha\in \Omega_{\omega}\setminus\{\min \Omega_{\omega}\}$ such that $\underline{S}_{\alpha}\neq \underline{S}_{\beta}$ for each $\beta<^{*}\alpha$, $M^{\alpha}=\min\limits_{\beta<^{*}\alpha}\underline{S}_{\beta}- \underline{S}_{\alpha}$. It can be easily checked that, by construction, $\Omega_{\med}=\Omega_{\omega}$ and the fixed locations $\gamma_f^{1},\ldots,\gamma_f^{a+1}$ satisfy all conditions in Definition \ref{mixedmedian}.In addition, by construction, it happens that for each $\alpha\in \Omega_{\omega}$, $\underline{M}^{\alpha}+\underline{S}^{\alpha}=a+1$.\medskip

\item We now show that, given a profile $R\in\mathcal{R}$, the outcome of $\omega(p(R))$ for the given type-anonymous left coalition system coincides with the outcome of $\med(p(R))$ for the collection of fixed locations defined.\smallskip

\noindent Let $R\in{\cal R}$ such that $\omega(p(R))=\alpha$. By Definition \ref{voterscheme}, $\omega(p(R))=\alpha$ if and only if $ \{i \in A: p(\Omega_i) \leq^* \alpha\} \in \mathcal{L}(\alpha)$ and for each $\beta \in \Omega_{\omega}$ such that $\beta <^* \alpha$, $\{i \in A: p(\Omega_i) \leq^* \beta\}\notin \mathcal{L}(\beta)$. Then, we have that $|\{i \in A: p(\Omega_i) \leq^* \alpha\}|\geq \underline{S}_{\alpha}$ and $|\{i \in A: p(\Omega_i) \leq^* \beta\}|< \underline{S}_{\beta}\geq \underline{S}_{\alpha}$.\smallskip

\noindent By construction, $\underline{M}^{\alpha}=(a+1)-\underline{S}_{\alpha}$ and $\underline{M}^{\beta}=(a+1)-\underline{S}_{\beta}$. Therefore, 
\begin{equation}\label{eq3}
|\{i \in A: p(\Omega_i) \leq^* \alpha\}|+ \underline{M}^{\alpha}\geq \underline{S}_{\alpha}+(a+1)-\underline{S}_{\alpha}=(a+1)
\end{equation}

and for each $\beta \in \Omega_{\omega}$ such that $\beta <^* \alpha$,
\begin{equation}\label{eq4} 
|\{i \in A: p(\Omega_i) \leq^* \beta\}|+ \underline{M}^{\beta}< \underline{S}_{\beta} + (a+1)- \underline{S}_{\beta}=(a+1).
\end{equation}

\noindent Hence, inequalities (\ref{eq3}) and (\ref{eq4}) together imply that $\med(p(R))=\alpha$. If $\alpha\in \Omega_f$, $f^{T1}(R)=\med(p(R))=\alpha=\omega(p(R))=f^{T2}(R)$. Otherwise, we go to the second step.
\end{itemize}

\noindent \emph{\underline{Second step}}\smallskip
\begin{itemize}
\item We first show how, given a collection of type-anonymous left-decisive sets, a collection of double-quotas can be constructed.\smallskip

\noindent Let $\mathbf{p}\in\Omega_f^{a}$ with $\omega(\mathbf{p})\in \Omega_{ \omega}\cup\Omega_{f}^{C^{2}}$ and $W$ be a collection of type-anonymous left-decisive sets as in Definition \ref{winningdef}. Then, for each $S, S'\in W(\omega(\mathbf{p}))$ such that $(|S\cap A|,|S\cap D|)\neq (|S'\cap A|,|S'\cap D|)$, we define two double-quotas $(q^{A}_{\omega(\mathbf{p})},q^{D}_{\omega(\mathbf{p})})^{1}\equiv(|S\cap A|,|S\cap D|)$ and $(q^{A}_{\omega(\mathbf{p})},q^{D}_{\omega(\mathbf{p})})^{2}=(|S'\cap A|,|S'\cap D|)$.\footnote{Note that the defined double-quotas are all different by the requirement of minimality imposed in the type-anonymous left-decisive sets.} Since $W(\omega(\mathbf{p}))$ is formed by minimal coalitions, the defined double-quotas are also minimal. It can be easily checked that, by construction, the defined double-quotas satisfy all conditions in Definition \ref{doublequota}.

\item We now show that, given a profile $R\in\mathcal{R}$, the outcome of $g_{\omega(p(R))}$ for the given collection of type-anonymous left-decisive sets coincides with the outcome of $t_{\med(p(R))}$ for the collection of double-quotas constructed.\smallskip

\noindent Let $R\in\mathcal{R}$ such that $\omega(p(R))\in \Omega_{f}^{C^{2}}$, and assume w.l.o.g. that $f^{T2}(R)=\underline{\omega}(p(R))$. Then $g_{\omega(p(R))}(R)=\underline{\omega}(p(R))$ and by Definition \ref{votingdef}, there is a coalition $S\in W(\omega(p(R)))$ such that $L_{\omega(p(R))}(R)\supseteq S$. By construction, there is a double-quota $q_{\med(p(R))}$ such that $q^{A}_{\med(p(R))}=|S\cap A|$ and $q^{D}_{\med(p(R))}=|S\cap D|$. Hence, $(|L^{A}_{\omega(p(R))}(R)|, |L^{D}_{\omega(p(R))}(R)|)\geqq (q^{A}_{\med(p(R))},q^{D}_{\med(p(R))})$ and by Definition \ref{doublequotamethod}, $t_{\med(p(R))}(R)=\underline{\med}(p(R))$ and $f^{T1}(R)=\underline{\med}(p(R))=\underline{\omega}(p(R))=f^{T2}(R)$.\medskip
\end{itemize}

\noindent The equivalence can be also checked with Examples~\ref{ex1} and \ref{ex3} for the first step and with Examples~\ref{ex2} and \ref{ex4} for the second one.\medskip

\noindent Finally, Theorem~\ref{theorem1} is proven by this equivalence.

\section{Concluding Remarks}\label{sec6}

\noindent This paper characterizes all strategy-proof and type-anonymous rules on a domain of single-peaked and single-dipped preferences where the type of preference of each agent is known but the location of the peak or dip and the rest of the preference are private information. The first characterization generalizes existing results on the single-peaked preference domain and the case of two alternatives. This result unfolds in two steps as follows: In the first step, we compute the median between the peaks and a fixed collection of locations, which can be single alternatives or pairs of contiguous alternatives. If the outcome of the median is a single alternative, then that is the final outcome of the rule. Otherwise, in the second step, we choose between the two alternatives of the pair using a double-quota majority method. While \cite{moulin1980strategy} previously characterized the strategy-proof and anonymous rules in the single-peaked preference domain using a median, we impose additional restrictions on the feasible fixed locations. Furthermore, \cite{moulin1983strategy} characterized strategy-proof and anonymous rules in the two-alternative case as ``quota majority methods." In contrast, we define a double-quota method to account for a society partitioned into two different types of preferences .\medskip

\noindent We also present a second characterization based on the findings in \cite{alcalde2024strategy}, which characterized all strategy-proof rules on the same domain. Building on their results, we analyze the additional restrictions imposed by type-anonymity. We find that the key distinction lies in the number of agents of each type supporting the outcomes at each step, rather than the specific composition of these ``supportive" coalitions. We finally provide an equivalence of the two characterizations that is also illustrated by examples throughout the paper.\medskip

\noindent This model does not accommodate indifferences in agents' preferences. The literature has explored rules that allow for indifferences in both steps. In the domain of single-peaked preferences, insights provided by \cite{moulin1980strategy} remain applicable, as indifferences do not affect agents' peaks. Strategy-proof and anonymous rules that allow for indifferences in the case of two alternatives have been examined in \cite{lahiri2020strategy,basile2020binary, basile2021structure,basile2022anonymous}. The latter characterized these rules as ``extended quota majority methods". However, introducing indifferences into our model may lead to difficulties.  Since we consider the peaks of the agents on the set $\Omega_f$, instead of on the set $X$, allowing indifferences could make agents be indifferent between two alternatives in $\Omega_f$. Consequently, single-peaked and single-dipped preferences on $X$ would become single-plateau \citep{moulin1984generalized, berga1998strategy} and single-basined \citep{bossert2014single} preferences on $\Omega_f$, respectively. In such situations, it becomes unclear whether the same two-step procedure remains viable. Addressing this challenge constitutes a promising approach for further research.

\bibliographystyle{ecta}
\bibliography{anonimidad}

\end{document}